\documentclass[sigconf]{acmart} 
\usepackage{graphicx}
\usepackage{amsmath}
\usepackage{subcaption}
\usepackage{epstopdf}
\usepackage{amsthm}
\usepackage{booktabs}
\usepackage{hyperref}
\usepackage{multirow}
\usepackage{romannum}
\usepackage{tabularx}

\usepackage{algorithm,algpseudocode}
 
\usepackage{amsfonts}

\newcommand{\RNum}[1]{\uppercase\expandafter{\romannumeral #1\relax}}

\setcopyright{none}

\acmDOI{00.000/000_0}
 
\acmISBN{000-0000-00-000/00/00}
 
\acmConference[0000]{000}{0000}{Beijing, China} 
\acmYear{0000}
\copyrightyear{0000}
\acmPrice{00.00}

\begin{document}
	\title{Solving the all pairs shortest path problem after minor update of a large dense graph}
	
	\author{Gangli Liu}
	\affiliation{%
		\institution{Tsinghua University}
	}
	\email{gl-liu13@mails.tsinghua.edu.cn}

\begin{abstract}
The all pairs shortest path problem is a fundamental optimization problem in graph theory. We deal with re-calculating the all-pairs shortest path (APSP) matrix after a minor modification of a weighted dense graph, e.g., adding a node, removing a node, or updating an edge. We assume the APSP matrix for the original graph is already known. The graph can be directed or undirected. A cold-start calculation of the new APSP matrix by traditional algorithms, like the Floyd-Warshall algorithm or Dijkstra's algorithm, needs $ O(n^3) $ time. We propose two algorithms for warm-start calculation of the new APSP matrix. The best case complexity for a warm-start calculation is $ O(n^2) $, the worst case complexity is $ O(n^3) $. We implemented the algorithms and tested their performance with experiments. The result shows a warm-start calculation can save a great portion of calculation time, compared with cold-start calculation.  In addition, another algorithm is devised to warm-start calculate of the shortest path between two nodes. Experiment shows warm-start calculation can save 99\% of calculation time, compared with cold-start calculation by Dijkstra's algorithm, on directed complete graphs of large sizes.

\end{abstract}
 
	\maketitle
 
\section{Introduction} \label{sec_one}

The Shortest Path Problem is a fundamental optimization problem in graph theory and computer science. It involves finding the shortest path between two vertices in a graph such that the sum of the weights of its constituent edges is minimized.

Let \( G = (V, E) \) be a graph where:

\begin{itemize}
	\item \( V \) is the set of vertices (nodes),
	\item \( E \subseteq V \times V \) is the set of edges (connections between nodes),
	\item \( w: E \to \mathbb{R}^+ \cup \{0\} \) is a weight function assigning a non-negative weight to each edge.
\end{itemize}

For a given source vertex \( s \in V \) and target vertex \( t \in V \), the Shortest Path Problem seeks to find a path \( P \) from \( s \) to \( t \) such that:

\[
P = \{ v_1, v_2, \dots, v_k \}, \quad v_1 = s, \, v_k = t,
\]

and the total path weight is minimized:

\[
\text{minimize: } W(P) = \sum_{i=1}^{k-1} w(v_i, v_{i+1}),
\]

where \( (v_i, v_{i+1}) \in E \) for \( i = 1, 2, \dots, k-1 \).

The all-pairs shortest path (APSP) problem compute the shortest paths between all pairs of vertices \( u, v \in V \). A dense graph is a graph in which the number of edges is close to the maximum possible number of edges for the given number of vertices. 

Let \( G = (V, E) \) be a graph, where \( |V| = n \) is the number of vertices and \( |E| \) is the number of edges. A dense graph satisfies:

\[
|E| \approx O(n^2)
\]

This means the number of edges grows quadratically with the number of vertices. For an undirected graph, the maximum possible number of edges is:

\[
\binom{n}{2} = \frac{n(n-1)}{2}.
\]

For a directed graph, the maximum possible number of edges is:

\[
n(n-1).
\]

A graph is considered dense when \( |E| \) is close to these upper bounds. Dense graphs are common in applications like social networks, transportation networks, or communication networks where most entities are interconnected. 

In this paper, we deal with re-calculating the all-pairs shortest path (APSP) matrix after a minor modification of a weighted dense graph, e.g., adding a node, removing a node, or updating an edge. We assume the APSP matrix for the original graph is already known. The graph can be directed or undirected. A straightforward method for calculating the APSP matrix of the updated graph is to use the Floyd–Warshall algorithm to recalculate the updated graph, it needs $ O(n^3) $ time. It is a very expensive time cost for a large dense graph. We are trying to utilize the already calculated APSP matrix to make calculation of the new APSP matrix less expensive.

\section{RELATED WORK} \label{Sec_related}

The Shortest Path Problem (SPP) is a foundational topic in graph theory and optimization, with numerous applications in transportation networks, telecommunications, and logistics \cite{ahuja1990faster,zhang2024mapreduce,marcucci2024shortest,zhu2024net,chen2024shortest}. Over the years, various algorithms and techniques have been developed to solve different variants of the problem efficiently.

\subsection{Classical Algorithms}
One of the earliest contributions to SPP dates back to the work of Dijkstra (1959), who proposed a greedy algorithm to solve the single-source shortest path problem for graphs with non-negative edge weights in \(O(|V|^2)\) time, later optimized to \(O(|E| + |V| \log |V|)\) using priority queues \cite{dijkstra1959note}. The idea of this algorithm is also given in (Leyzorek et al. 1957) \cite{leyzorek1957investigation}. 

For graphs with negative edge weights, the Bellman-Ford algorithm (1958) provides a reliable solution, albeit with a higher computational cost of \(O(|V||E|)\) \cite{bellman1958routing}. The Floyd-Warshall algorithm (1962) extends these techniques to compute all-pairs shortest paths in \(O(|V|^3)\), leveraging dynamic programming \cite{floyd1962algorithm}.

\subsection{Optimizations and Modern Variants}
Advances in data structures, such as Fibonacci heaps \cite{fredman1987fibonacci}, have further improved the efficiency of Dijkstra’s algorithm. More recently, heuristic-based approaches like $A^{*}$ have been widely adopted for real-world applications, where an admissible heuristic guides the search to improve runtime performance \cite{hart1968formal}. Parallel and distributed versions of shortest path algorithms have also emerged, leveraging modern computing architectures for large-scale graphs \cite{meyer1998delta}.

\subsection{Specialized Applications}
The SPP has been extended to address specialized scenarios, such as the multi-criteria shortest path problem, which considers trade-offs between multiple objectives, like cost and time \cite{hansen1980bicriterion}. In dynamic or time-dependent graphs, the edge weights may vary over time, necessitating new algorithms like the time-expanded shortest path \cite{cooke1966shortest}. Additionally, the rise of massive graphs in social networks and geographic information systems has spurred the development of approximate methods, such as graph sparsification and sketching \cite{cohen1997size}.

\subsection{Challenges in Dense and Weighted Graphs}
For dense graphs, where the number of edges approaches \(O(|V|^2)\), naive algorithms often become computationally expensive. Techniques like matrix-based methods for all-pairs shortest paths \cite{johnson1977efficient} or GPU-accelerated implementations \cite{katz2008all} have shown promise in reducing computational overhead.

\subsection{Emerging Trends}
Recent research has explored incorporating machine learning into shortest path computations. These methods predict likely paths or edge weights, complementing traditional algorithms in scenarios with incomplete or noisy data \cite{velivckovic2017graph}. Moreover, shortest path calculations are increasingly being integrated with clustering and community detection tasks to solve problems in network science and biology \cite{newman2003structure}.

\begin{table*}
	\caption{Table of notations}
	\begin{tabularx}{0.95\textwidth}{@{}XX@{}}
		\toprule
		$G$ & A weighted dense graph of N nodes, with each node indexed from 1 to N. Graph G is supposed to be directed, an undirected graph can be considered as a special case of a directed graph;\\ \midrule
		$G_{[1,n]}$ &A graph that is composed of the first $ n $ nodes of  $G$,  the nodes are indexed from 1 to n;\\ \midrule
		
		$G_{n+1}$ &The $( n+1) $th node of  $G$;\\\midrule
		$G$ + p &  Graph $G$ plus one new node $ p $. Since $ p \notin G $, if $G$ has N nodes, this new graph now has $ N + 1 $ nodes;\\\midrule
		
		$G - G_k$   &  A new graph by removing the kth node from graph $G$. If $G$ has N nodes, this new graph now has $ N - 1 $ nodes;\\\midrule
		
		$G^{'}$   &  A new graph by modifying weight of  one edge of graph $G$, or by removing a node from $ G $, or by adding a node to $ G $;\\\midrule
		
		$\Psi_{(i,j,n,G)}$  & $\Psi_{(i,j,n,G)}$ is a sequence from node i to node j, which has a total number of n nodes. All the nodes in the sequence must belong to graph $ G $. That is to say, it is a path starts from i, and ends with j. The path is not allowed to have loops, unless the start and the end is the same node;  \\\midrule

		$d(i,j)$  & $d(i,j)$ is the adjacency distance from node i to node j on graph G. Note the graph is directed;  \\\midrule
				
		$len(~\Psi_{(i,j,n,G)}~)$	  & $len(~\Psi_{(i,j,n,G)}~)$ is the length of path $\Psi_{(i,j,n,G)}$, which is the sum of edge weights on the path; \\\midrule
		
		$\Theta_{(i,j,G)}$	  & $\Theta_{(i,j,G)}$	 is the set of all paths from node i to node j. A path in $\Theta_{(i,j,G)}$ can have arbitrary number of nodes (at least two). All the nodes in a path must belong to graph $ G $; \\   \midrule
		
		$SPD(i,j~|~G)$ & 	 $SPD(i,j~|~G)$ is the shortest path distance (SPD) from node i to node j, where $ G $ is the   $ \boldsymbol{Context} $ of the shortest path distance. The $ \boldsymbol{Context} $ of a distance  is defined in \cite{liu2023min}. A node's shortest path distance to itself is always 0; \\ \midrule
	
		$ \mathbb{M}_{k,G_{[1,k]}} $ & $ \mathbb{M}_{k,G_{[1,k]}} $ is the pairwise shortest path distance matrix of $ G_{[1,k]} $, which has shape $ k \times k $. The shortest path distances are under the $ \boldsymbol{Context} $ of $ G_{[1,k]} $; \\ \midrule
		
		$ \mathbb{M}_{G} $ & The APSP matrix of $ G $, $ \mathbb{M}_{G}  = \mathbb{M}_{N,G_{[1,N]}} $; \\ 
 
		\bottomrule
	\end{tabularx}
\label{tab:notations}
\end{table*}

\section{Updating a large graph}
In a previous paper, we propose Algorithm 1 (MMJ distance by recursion) for solving the all pairs minimax path problem or widest path problem \cite{liu2023min}. It can also be revised to solve the APSP matrix of the shortest path problem, which also takes $ O(n^3) $ time.

\subsection{APSP after adding a node}
As discussed in Section 6.1 (Merit of Algorithm 1) of \cite{liu2024efficient},
Algorithm 1 (MMJ distance by recursion) has the advantage of warm-start capability. Consider the scenario where we have already computed the APSP matrix $ \mathbb{M}_{G} $ for a large graph $ G $, and a new point or node, $ p $, which is not part of $ G $, is introduced. The updated graph is referred to as $ G + p $. When determining the APSP matrix for $ G + p $, conventional algorithms like Floyd–Warshall algorithm or Dijkstra's algorithm might necessitate a computation beginning from scratch, which takes $ O(n^3) $ time. 

Algorithm 1 leverages the precomputed $ \mathbb{M}_{G} $ to facilitate the calculation of the new APSP matrix of graph $ G + p $, in accordance with the results of Theorem \ref{theorem2}, \ref{theo_2_back}, \ref{theorem3}, and Corollary \ref{corollary1}, \ref{corollary2}, which are revised from the theorems and corollary in Section 3.3 (Other properties of MMJ distance) of \cite{liu2023min}.  A warm-start of Algorithm 1 requires only $ O(n^2) $ time, which is much less expensive than a cold-start of conventional algorithms, which takes $ O(n^3) $ time.

\begin{theorem} 
	\label{theorem2}
	Suppose $ r \in \{1,2, \dots, n \}$,
	
	\begin{equation}
		f(t) =   d(G_{n+1},G_t) + SPD(G_t,G_r~|~G_{[1,n]})   
	\end{equation}		
	\begin{equation}
		\mathbb{X} =  \{        f(t)   ~ | ~  t \in \{1,2, \dots, n \}   \}
	\end{equation}	
	then, 	
	\begin{equation}
		SPD(G_{n+1},G_r~|~G_{[1,n+1]})  = min(\mathbb{X} )
	\end{equation}
	
\end{theorem}

For the meaning of $ G_t,G_r,G_{[1,n]},   and ~  G_{[1,n+1]}$, see Table \ref{tab:notations}.

\begin{proof}
	There are $ n $ possibilities of the shortest path from $ G_{n+1} $ to $ G_r $, under the context of $ G_{[1,n+1]} $, set $ \mathbb{X} $ enumerate them all. Each element of $ \mathbb{X} $ is the shortest path distance of each possibility. Therefore, according to the definition of shortest path distance, $ SPD(G_{n+1},G_r~|~G_{[1,n+1]})  = min(\mathbb{X} ) $. The $ n $ possibilities are not mutually exclusive; multiple possibilities can happen simultaneously if the shortest paths are not unique.
\end{proof}

\begin{theorem} 
	\label{theo_2_back}
	Suppose $ r \in \{1,2, \dots, n \}$,
	
	\begin{equation}
	f(t) = SPD(G_r,G_t~|~G_{[1,n]}) +  d(G_t, G_{n+1})
	\end{equation}		
	\begin{equation}
	\mathbb{X} =  \{        f(t)   ~ | ~  t \in \{1,2, \dots, n \}   \}
	\end{equation}	
	then, 	
	\begin{equation}
	SPD(G_r, G_{n+1}~|~G_{[1,n+1]})  = min(\mathbb{X} )
	\end{equation}
	
\end{theorem}
\begin{proof}
	The proof is similar to proof of Theorem \ref{theorem2}. Since we are dealing with a directed graph, the order of nodes in a distance notation matters.
\end{proof}

\begin{corollary}
	\label{corollary1}
	Suppose $ r \in \{1,2, \dots, N \}, p \notin G$,	
	\begin{equation}
		f(t) =  d(p,G_t) + SPD(G_t,G_r~|~G)   
	\end{equation}		
	\begin{equation}
		\mathbb{X} =  \{        f(t)   ~ | ~  t \in \{1,2, \dots, N \}   \}
	\end{equation}		
	then,	
	\begin{equation}
		SPD(p,G_r~|~G + p)  = min(\mathbb{X} )
	\end{equation}
\end{corollary}
For the meaning of $ G + p $, see Table \ref{tab:notations}.
\begin{proof}
	The proof follows the conclusion of Theorem \ref{theorem2}.
\end{proof}

\begin{corollary}
	\label{corollary2}
	Suppose $ r \in \{1,2, \dots, N \}, p \notin G$,	
	\begin{equation}
	f(t) = SPD(G_r, G_t~|~G) + d(G_t, p)    
	\end{equation}		
	\begin{equation}
	\mathbb{X} =  \{        f(t)   ~ | ~  t \in \{1,2, \dots, N \}   \}
	\end{equation}		
	then,	
	\begin{equation}
	SPD(G_r, p~|~G + p)  = min(\mathbb{X} )
	\end{equation}
\end{corollary}
\begin{proof}
	The proof follows the conclusion of Theorem \ref{theo_2_back}.
\end{proof}

\begin{theorem} 
	\label{theorem3}
	Suppose $ i,j \in \{1,2, \dots, n \}$,	
	\begin{equation}
		x_1 = SPD(G_{i},G_j~|~G_{[1,n]})
	\end{equation}		
	\begin{equation}
		t_1 = SPD(G_i, G_{n+1}~|~G_{[1,n+1]})
	\end{equation}		
	\begin{equation}
		t_2 = SPD(G_{n+1},G_j~|~G_{[1,n+1]})
	\end{equation}	
	\begin{equation}
		x_2 = t_1 +  t_2
	\end{equation}		
	then,	
	\begin{equation}
		SPD(G_{i},G_j~|~G_{[1,n+1]})  = min(x_1 , ~  x_2)
	\end{equation}
	
\end{theorem}
\begin{proof}
	If the shortest paths are not unique, there are two possibilities for the shortest path from $ G_{i} $ to $ G_j $, under the context of $ G_{[1,n+1]} $: 
	\begin{enumerate}
		\item There exists one shortest path from $ G_{i} $ to $ G_j $ which does not contain node $ G_{n+1} $, which means $ G_{n+1} $ is not necessary for the shortest path from $ G_{i} $ to $ G_j $, under the context of $ G_{[1,n+1]} $. That is to say:		
		\begin{equation*}
		SPD(G_{i},G_j~|~G_{[1,n+1]}) = SPD(G_{i},G_j~|~G_{[1,n]})
		\end{equation*}
		\item All the shortest paths from $ G_{i} $ to $ G_j $ must contain node $ G_{n+1} $, which means $ G_{n+1} $ is necessary for the shortest path  from $ G_{i} $ to $ G_j $, under the context of $ G_{[1,n+1]} $. That is to say:	
		\begin{equation*}
		SPD(G_{i},G_j~|~G_{[1,n+1]}) \neq SPD(G_{i},G_j~|~G_{[1,n]})
		\end{equation*}
	\end{enumerate}
	  $ x_1 $ is the SPD of the first possibility; $ x_2 $ is the SPD of the second possibility. Therefore, according to the definition of shortest path distance, $ SPD(G_{i},G_j~|~G_{[1,n+1]})  = min(x_1 , ~  x_2) $. If the shortest path is  unique, the reasoning still holds. The two possibilities are mutually exclusive; they cannot happen simultaneously.
\end{proof}

\begin{theorem} 
	\label{theorem_spd_4}
	Suppose $ i,j \in \{1,2, \dots, n \}$,	
		\begin{equation}
	t_1 = SPD(G_i, G_{n+1}~|~G_{[1,n+1]})
	\end{equation}		
	\begin{equation}
	t_2 = SPD(G_{n+1},G_j~|~G_{[1,n+1]})
	\end{equation}	
	\quad if,
	\begin{equation}
	SPD(G_{i},G_j~|~G_{[1,n+1]}) < t_1 +  t_2
	\label{theo_4_if}
	\end{equation}		
 	
	then,	
	\begin{equation}
	SPD(G_{i},G_j~|~G_{[1,n+1]}) = SPD(G_{i},G_j~|~G_{[1,n]})
	\label{theo_4_conclu}
	\end{equation}
	
	which means $ G_{n+1} $ is not necessary for the shortest path from $ G_{i} $ to $ G_j $, under the context of $ G_{[1,n+1]} $. 
	
\end{theorem}
\begin{proof}
As discussed in the proof of Theorem \ref{theorem3}, there are two possibilities for the shortest path from $ G_{i} $ to $ G_j $, under the context of $ G_{[1,n+1]} $. And the two possibilities are mutually exclusive; they cannot happen simultaneously. We only need to negate possibility $ (2) $, then we can arrive to the conclusion of Equation \ref{theo_4_conclu}. Suppose possibility $ (2) $ happens, then $ G_{n+1} $ is necessary for the shortest path  from $ G_{i} $ to $ G_j $, under the context of $ G_{[1,n+1]} $. Then $ SPD(G_{i},G_j~|~G_{[1,n+1]}) = t_1 +  t_2 $, which is contradicted to Equation \ref{theo_4_if}. Therefore, possibility $ (2) $ cannot happen; only possibility $ (1) $ can happen. 
\end{proof}

\begin{corollary}
	\label{corollary3}
Suppose $ i,~j,~k\in \{1,2, \dots, N \}$,	$ k \neq i$, $ k \neq j $,
\begin{equation}
t_1 = SPD(G_i, G_{k}~|~G)
\end{equation}		
\begin{equation}
t_2 = SPD(G_{k},G_j~|~G)
\end{equation}	
\quad if,
\begin{equation}
SPD(G_{i},G_j~|~G) < t_1 +  t_2
\end{equation}		

then,	
\begin{equation}
SPD(G_{i},G_j~|~G) = SPD(G_{i},G_j~|~G - G_{k})
\end{equation}

which means $ G_{k} $ is not necessary for the shortest path from $ G_{i} $ to $ G_j $, under the context of $ G$. 
\end{corollary}

\begin{proof}
	The proof follows the conclusion of Theorem \ref{theorem_spd_4}. We just re-index the nodes in graph $ G $. 
\end{proof}

\subsection{APSP after removing a node}
Sometimes, we need to remove a node from a large graph. Suppose we removed the kth node from graph $ G $, the new graph is noted $G - G_k$ (Table \ref{tab:notations}). A cold-start of conventional algorithms for calculating the APSP matrix of graph $G - G_k$ will take $ O(n^3) $ time. 

However, we can take a smarter method to make the computation less expensive. Firstly, we make a $ need\_update\_list $ to record which nodes' shortest path distance (SPD)  to others are affected by the deletion. E.g., in Figure \ref{fig:graph1}, if we delete Node C, we get Figure \ref{fig:graph2} (since the matrices are symmetric, we only show half of them). 
Removing a node is equivalent to set the node's distances (to and from)  to other nodes to infinity. The $ need\_update\_list $ is totally empty for Figure \ref{fig:graph2}. Because none of the pair-wise shortest path distances are affected by removing Node C. Except Node C itself. E.g., 
\begin{align*} 
SPD(A, B~|~G) &= SPD(A, B~|~G - C) \\ 
SPD(A, D~|~G) &= SPD(A, D~|~G - C)\\
SPD(B, A~|~G) &= SPD(B, A~|~G - C) \\ 
SPD(B, D~|~G) &= SPD(B, D~|~G - C)\\
SPD(D, A~|~G) &= SPD(D, A~|~G - C) \\ 
SPD(D, B~|~G) &= SPD(D, B~|~G - C)
\end{align*}

The $ need\_update\_list $ for Figure \ref{fig:graph2} looks like this:
\begin{align*} 
Node ~ A &: empty\\
Node ~ B &: empty\\
Node ~ D &: empty
\end{align*}

In Figure \ref{fig:graph3}, we removed Node B. Some of the remaining pair-wise shortest path distances are affected, some are not. E.g.,
\begin{align*} 
SPD(A, C~|~G) &\neq SPD(A, C~|~G - B) \\ 
SPD(A, D~|~G) &\neq SPD(A, D~|~G - B)\\
SPD(C, A~|~G) &\neq SPD(C, A~|~G - B) \\ 
SPD(C, D~|~G) &= SPD(C, D~|~G - B)\\
SPD(D, A~|~G) &\neq SPD(D, A~|~G - B) \\ 
SPD(D, C~|~G) &= SPD(D, C~|~G - B)
\end{align*}

The $ need\_update\_list $ for Figure \ref{fig:graph3} looks like this:
\begin{align*} 
Node ~ A &: ~ [A,~C], ~[A,~D]\\
Node ~ C &: ~ [C,~A]\\
Node ~ D &: ~ [D,~A]
\end{align*}

We use Theorem \ref{theorem_spd_4} to construct the $ need\_update\_list $, by setting the node to be removed as $ G_{n+1} $.
If
\begin{equation*}
SPD(G_{i},G_j~|~G_{[1,n+1]}) < t_1 +  t_2
\label{theo_4_if_less}
\end{equation*}	
which means $ G_{n+1} $ is not necessary for the shortest path from $ G_{i} $ to $ G_j $, under the context of $ G_{[1,n+1]} $, then $ G_{n+1} $ can be safely removed from the graph, without affecting the shortest path distance from $ G_{i} $ to $ G_j $. So, node pair $ [G_{i},~G_j] $ will not appear in the $ need\_update\_list $. Otherwise, if
\begin{equation*}
SPD(G_{i},G_j~|~G_{[1,n+1]}) = t_1 +  t_2
\label{theo_4_if_equal}
\end{equation*}	
then we are not sure whether the shortest path distance from $ G_{i} $ to $ G_j $ will be affected by removing node $ G_{n+1} $ from the graph; the SPD from $ G_{i} $ to $ G_j $ needs to be re-calculated after the removing. So, node pair $ [G_{i},~G_j] $ will be appended to the $ need\_update\_list $ of node $ G_{i} $. Corollary \ref{corollary3} makes it easier to understand than Theorem \ref{theorem_spd_4}. Constructing  the $ need\_update\_list $ only needs $ O(n^2) $ time.

	\begin{figure*} 
	\begin{subfigure}{0.33\textwidth}
	\includegraphics[width=\linewidth]{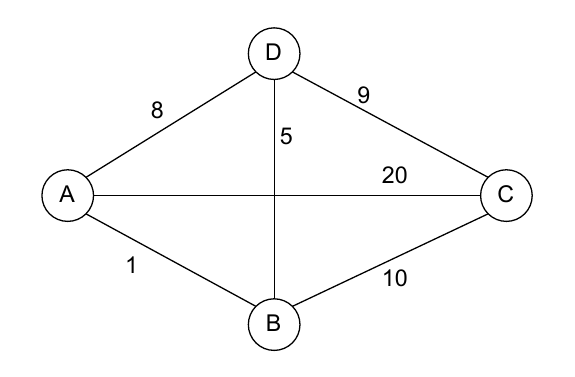}
	\caption{Graph}   \label{fig:graph11}
\end{subfigure}    
	\begin{subfigure}{0.33\textwidth}
		\includegraphics[width=0.8\linewidth]{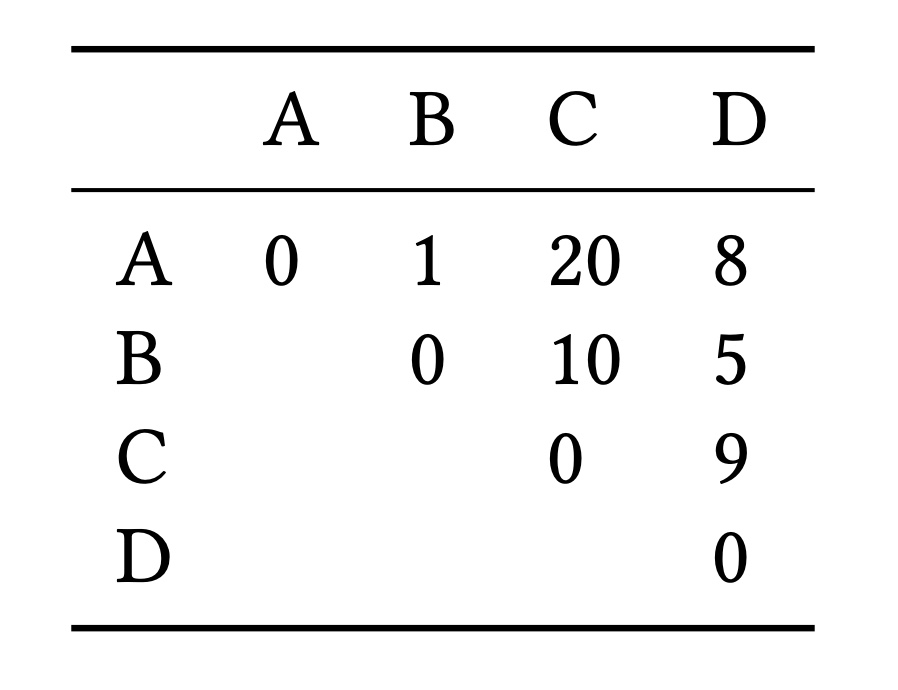}
		\caption{adjacency matrix}   \label{fig:graph12}
	\end{subfigure}    
	\begin{subfigure}{0.33\textwidth}
	\includegraphics[width=0.8\linewidth]{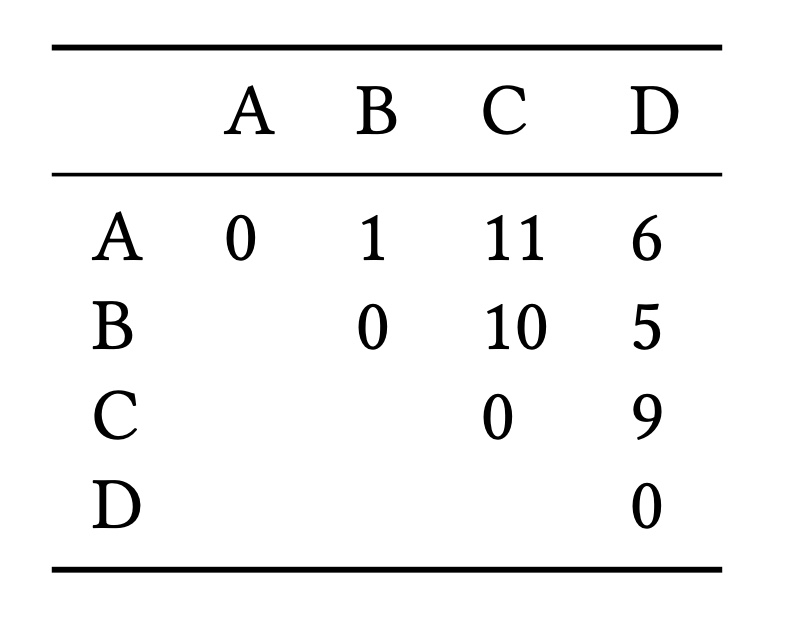}
	\caption{APSP matrix}   \label{fig:graph13}
\end{subfigure} 
 
	\caption{Graph, adjacency matrix, and APSP matrix.}
	 \label{fig:graph1}
\end{figure*}

	\begin{figure*} 
	\begin{subfigure}{0.33\textwidth}
	\includegraphics[width=\linewidth]{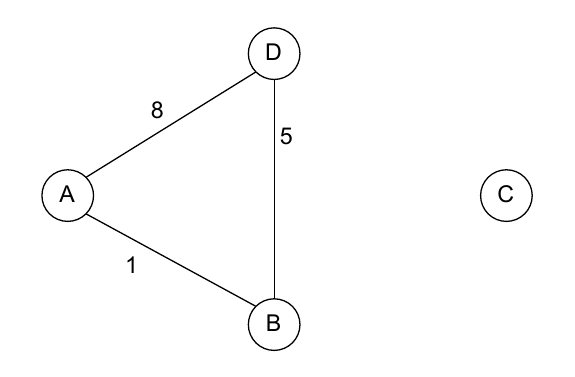}
	\caption{Graph}   \label{fig:graph21}
\end{subfigure}    
	\begin{subfigure}{0.33\textwidth}
		\includegraphics[width=0.8\linewidth]{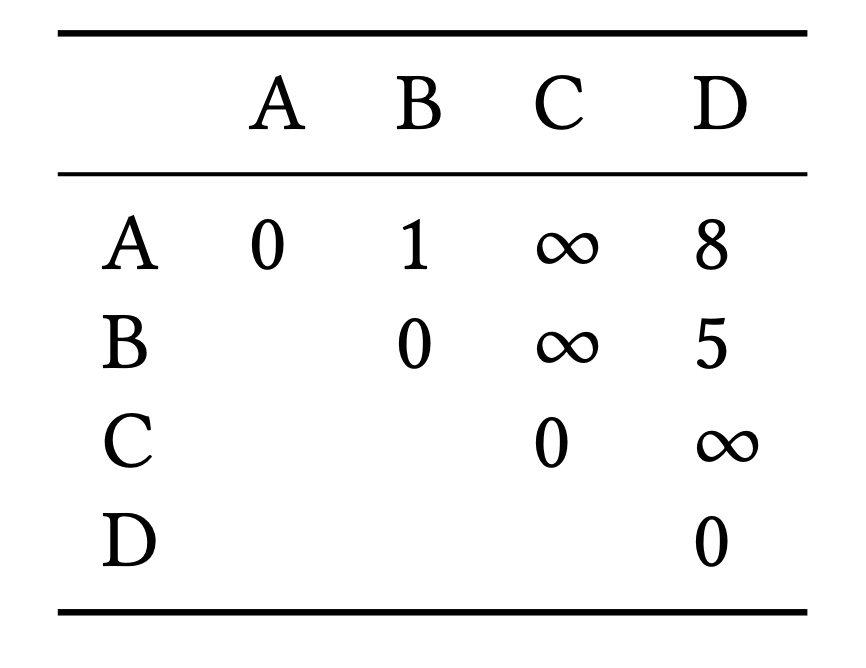}
		\caption{adjacency matrix}   \label{fig:graph22}
	\end{subfigure}    
	\begin{subfigure}{0.33\textwidth}
	\includegraphics[width=0.8\linewidth]{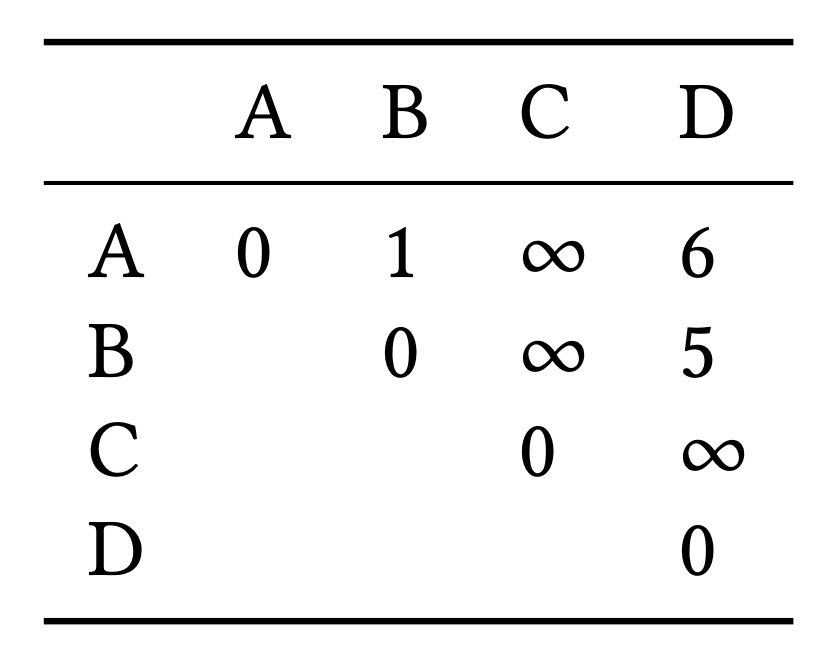}
	\caption{APSP matrix}   \label{fig:graph23}
\end{subfigure} 
 
	\caption{Graph, adjacency matrix, and APSP matrix, after removing node C.}
	 \label{fig:graph2}
\end{figure*}

	\begin{figure*} 
	\begin{subfigure}{0.33\textwidth}
	\includegraphics[width=\linewidth]{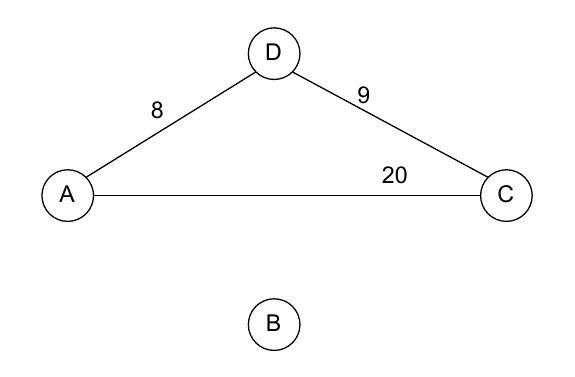}
	\caption{Graph}   \label{fig:graph31}
\end{subfigure}    
	\begin{subfigure}{0.33\textwidth}
		\includegraphics[width=0.8\linewidth]{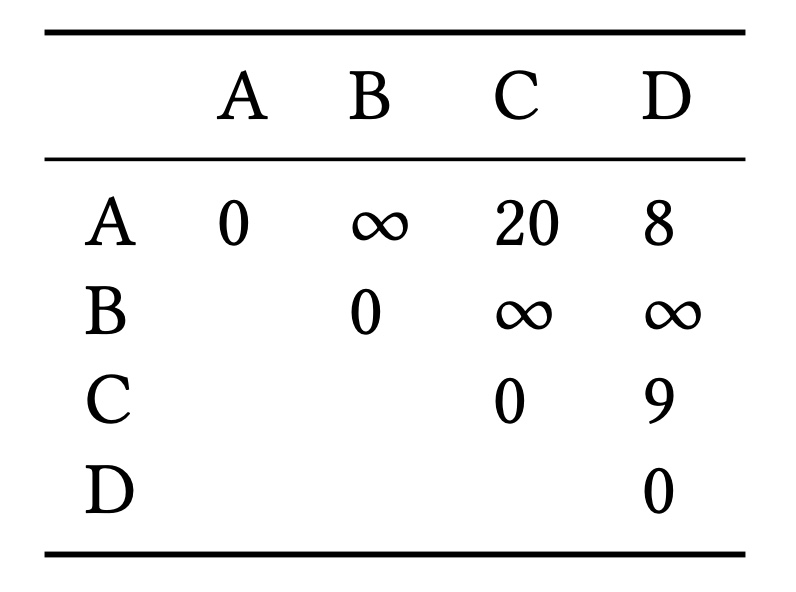}
		\caption{adjacency matrix}   \label{fig:graph32}
	\end{subfigure}    
	\begin{subfigure}{0.33\textwidth}
	\includegraphics[width=0.8\linewidth]{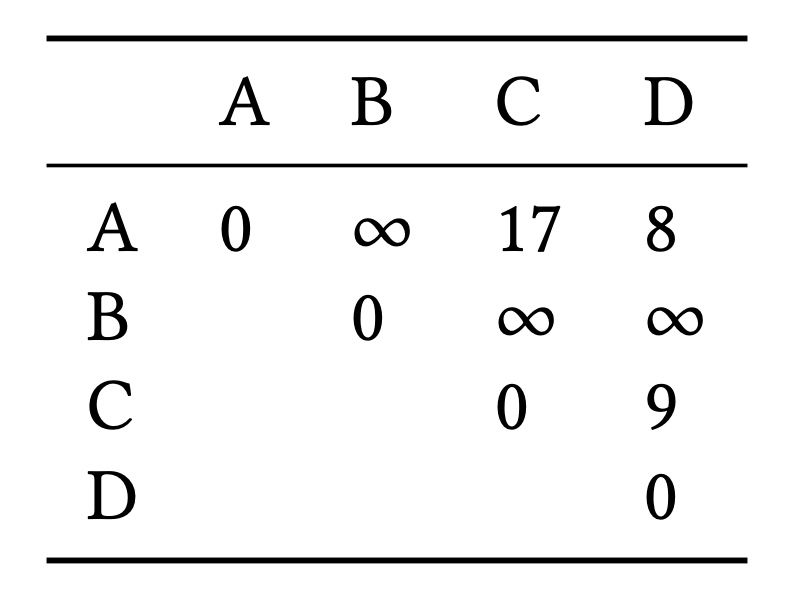}
	\caption{APSP matrix}   \label{fig:graph33}
\end{subfigure} 
 
	\caption{Graph, adjacency matrix, and APSP matrix, after removing node B.}
	 \label{fig:graph3}
\end{figure*}

\begin{definition}
	Cost for calculating the new APSP matrix after removing node $ G_k$ from $ G $.

	\begin{equation}
	C(G, \mathbb{M}_{G}, G_k) =  \frac{\Phi(\mathbb{M}_{G}, G_k)}{N - 1}	
	\label{equ:cost_defi}
	\end{equation}
	Where $ \mathbb{M}_{G} $ is the APSP matrix of graph $ G $; $ G_k $ is the node to be removed; $ \Phi(\mathbb{M}_{G}, G_k) $ is the number of non-empty items in the $ need\_update\_list $, after removing node $ G_k$ from $ G $; $ N $ is the number of vertices in graph $ G $. The range of $ C(G, \mathbb{M}_{G}, G_k)  $ is $ [0, 1] $.
	\label{def:cost_defi}
\end{definition}

We devise an algorithm for solving the APSP matrix after removing a node. In Algorithm \ref{alg:APSPremove}, we first construct the $ need\_update\_list $ after removing node $ G_k$, then use the $ need\_update\_list $ to calculate the $ Cost $ defined in Definition \ref{def:cost_defi}. If the $ Cost $ is larger than hyper-parameter $ \delta $, we just use the Floyd-Warshall algorithm to re-calculate the APSP matrix of the new graph $G - G_k$. If the $ Cost $ is small, we use Dijkstra's algorithm to calculate a node's distance to other nodes in the new graph $ G - G_k $. Hopefully, only a few nodes will be affected by removing node $ G_k$, therefore, saved time for calculating the new APSP matrix. 

So, removing a node from a graph is harder than adding a node to the graph, for calculating the APSP matrix. Adding a node only needs $ O(n^2) $ time even for the worst case. For removing a node, the best case complexity is $ O(n^2) $, the worst case complexity is $ O(n^3) $. E.g., in the $ need\_update\_list $ for Figure \ref{fig:graph2}, all the items are empty, the complexity for solving the new APSP matrix is  $ O(n^2) $, which is used for calculating the $ need\_update\_list $. In the $ need\_update\_list $ for Figure \ref{fig:graph3}, all the items are non-empty, the complexity is $ O(n^3) $.
\setcounter{algorithm}{4}
\begin{algorithm} 
	\caption{APSP after removing a node}
	\begin{algorithmic}[1]
		\Require{$G$, APSP of G: $ \mathbb{M}_{G} $,  node to be removed: $ G_k $, hyper-parameter: $\delta$} 
		\Ensure{$APSP ~of ~G - G_k:  \mathbb{M}_{G - G_k} $}
		\Statex
		\Function{APSP\_remove\_node}{$G$, $ \mathbb{M}_{G} $, $ G_k $, $\delta$}
		
		\State {$remaining\_node\_list$ $\gets$ $G - G_k$}
		\State {$need\_update\_list$ $\gets$ $cal\_need\_update\_list(G_k, \mathbb{M}_{G})$}
		
		\State {$Cost$ $\gets$ $cal\_cost\_of\_remove(need\_update\_list)$}

		\If{$ Cost >  \delta $}
		\State {$\mathbb{M}_{G - G_k}$ $\gets$ $Floyd\_Warshall(G - G_k)$}
		\State \Return {$\mathbb{M}_{G - G_k} $}
		\EndIf	

		\State {$\mathbb{M}_{G - G_k}$ $\gets$ $copy(\mathbb{M}_{G})$}
		\For{$i$ $ in $ $remaining\_node\_list$}
			\If{$ len(need\_update\_list[i]) > 0 $}
			\State {//We can stop early if the shortest path tree has covered all the nodes in $need\_update\_list[i]$.} 			
				\State {$temp$ $\gets$ $dijkstra\_one\_to\_all\_others(G - G_k,~i)$} 
				\For{$j$ $ in $ $need\_update\_list[i]$}
					\State {$\mathbb{M}_{G - G_k}[i,j]$ $\gets$ $temp[j]$}  					
				\EndFor			 
			\Else
				\State {$pass$}
			\EndIf      
		\EndFor 
		\State \Return {$\mathbb{M}_{G - G_k} $}
		\EndFunction
	\end{algorithmic}
	\label{alg:APSPremove}
\end{algorithm}

\subsection{APSP after modifying  an edge}
Modifying  an edge can be accomplished by removing one of the edge's vertices, then add the node back, with the edge being updated. So,  the best case complexity for modifying  an edge is $ O(n^2) $, the worst case complexity is $ O(n^3) $; the same complexity as removing a node.

Algorithm \ref{alg:APSPmodify} is devised to calculate the APSP matrix after modifying  an edge. It firstly remove a node associated with the edge from the graph, then add the node back, with the edge being updated. An edge is associated with two nodes. So, before removing a node, it calculates which node is cheaper to remove, then remove the cheaper one. 
\begin{algorithm} 
	\caption{APSP after modifying  an edge}
	\begin{algorithmic}[1]
		\Require{$G$, APSP of G: $ \mathbb{M}_{G} $,  edge nodes: $ e_n$, edge weight: $ e_w $, hyper-parameter: $\delta$} 
		\Ensure{APSP of new graph $ G^{'} $:  $\mathbb{M}_{G^{'}} $}
		\Statex
		\Function{APSP\_modify\_edge}{$G$, $ \mathbb{M}_{G} $, $ e_n $, $ e_w $, $ \delta $}
		\State {$i$ $\gets$ $e_n[0]$}
		\State {$j$ $\gets$ $e_n[1]$}
		\State {$Cost\_i$ $\gets$ $cal\_cost\_of\_remove(G, i)$}
		\State {$Cost\_j$ $\gets$ $cal\_cost\_of\_remove(G, j)$}
		\If{$ Cost\_i  <   Cost\_j $}
		\State {$G_k$ $\gets$ $i$}
		\Else
		\State {$G_k$ $\gets$ $j$}
		\EndIf	
		\State {//Calculate APSP matrix after removing node.}			
		\State {$\mathbb{M}_{G - G_k}$ $\gets$ $APSP\_remove\_node(G, \mathbb{M}_{G}, G_k, \delta)$}
		
		\State {//Add the node back, update the edge, then use warm-start of Algorithm 1 (MMJ distance by recursion) to calculate the new APSP matrix.} 
		\State {$G^{'}$ $\gets$ $cal\_updated\_graph(G, e_n, e_w)$}
		\State {$ \mathbb{M}_{G^{'}}$ $\gets$ $APSP\_add\_node(G^{'}, \mathbb{M}_{G - G_k})$}		
	
		\State \Return {$\mathbb{M}_{G^{'}} $}
		\EndFunction
	\end{algorithmic}
	\label{alg:APSPmodify}
\end{algorithm}

\section{warm-start calculation of shortest path}
We can carry out a warm-start calculation of the shortest path between two nodes, based on the already known APSP matrix and the conclusion of Theorem \ref{theorem_spd_4}.

Algorithm \ref{alg:sp_apsp} is devised for warm-start calculation of the shortest path between two nodes, based on the APSP matrix. It use the conclusion of Theorem \ref{theorem_spd_4} to exclude unnecessary nodes from node $ i $ to node $ j $, generate the $ candidate\_node\_list $; then form a small graph which is composed of nodes in $ candidate\_node\_list $; then use Dijkstra's algorithm to calculate the shortest path from node $ i $ to node $ j $, on the small graph; then translate the path into original node index. 

Since the $ candidate\_node\_list $ is usually very small, calculating the shortest path from node $ i $ to $ j $ on the small graph usually needs only $ O(1) $ time. So the average case complexity of Algorithm \ref{alg:sp_apsp} is $ O(n) $.

\begin{algorithm} 
	\caption{warm-start calculation of shortest path}
	\begin{algorithmic}[1]
	\Require{APSP of G: $ \mathbb{M}_{G} $,  Adjacency matrix: $ \mathbb{A}_{G} $,  start node: $ i$,  end node: $ j$}
		\Ensure{Shortest path from $ i $ to $ j $: $ path(i,~j) $}
		\Statex
		\Function{warm\_cal\_shortest\_path}{$ \mathbb{M}_{G} $, $ \mathbb{A}_{G} $, $ i$, $j$}
		\If{$ i == j $ }
		\State \Return {$[i] $}
		\EndIf
 
		\State {$remaining\_node\_list$ $\gets$ $G - i - j$}
		\State {$candidate\_node\_list$ $\gets$ $empty\_list$}
		\State {$candidate\_node\_list.append(i)$}
		\For{$t$ $ in $ $remaining\_node\_list$}
		\If{$ \mathbb{M}_{G}[i,~j]  <  \mathbb{M}_{G}[i,~t]  + \mathbb{M}_{G}[t,~j] $ }
		\State {$pass$}
		\Else
		\State {$candidate\_node\_list.append(t)$}
		\EndIf			
		\EndFor	
		\State {$candidate\_node\_list.append(j)$}
		\State {$K$ $\gets$ $len(candidate\_node\_list)$}
		
		\State {$small\_matrix$ $\gets$ $zeros((K, K))$}
		\For{$i,~m$ $ in $ $enumerate(candidate\_node\_list)$}
		\For{$j,~n$ $ in $ $enumerate(candidate\_node\_list)$}
		\State {$small\_matrix[i,~j]$ $\gets$ $\mathbb{A}_{G}[m,~n]$}
		\EndFor	
		\EndFor	
		\State {//Use Dijkstra's algorithm to calculate the path from node $ 0 $ to node $ K-1 $, on the graph defined by $small\_matrix$.}
		
		\State {$path$ $\gets$ $cal\_path\_by\_dijkstra(small\_matrix, 0, K-1)$}
		\State {//Translate the path into original node index.}	
		\State {$ path(i,~j) $ $\gets$ $[candidate\_node\_list[i]~ for ~i ~in~ path]$}

		\State \Return {$ path(i,~j) $}
		\EndFunction
	\end{algorithmic}
	\label{alg:sp_apsp}
\end{algorithm}
\subsection{Correctness proof of Algorithm \ref{alg:sp_apsp}}
The correctness of Algorithm \ref{alg:sp_apsp} follows the conclusion of Theorem \ref{al_correct}.
\begin{theorem} 
	\label{al_correct}
	The small graph which is composed of nodes in $ candidate\_node\_list $ in Algorithm \ref{alg:sp_apsp} contains all the shortest paths from node $ i $ to  $ j $ on graph $ G $.
\end{theorem}
\begin{proof}
We can divide  nodes  in graph $ G $  into two sets: nodes in  $ candidate\_node\_list $, noted $G_c$;  nodes not in  $ candidate\_node\_list $, noted  $G_r$. Suppose there exists a shortest path from node $ i $ to  $ j $ on graph $ G $ involves a node in $G_r$, the involved node is noted $ \xi $. The path is noted $ p(i,~\xi)  + p(\xi,~j)$. The APSP matrix of G is $ \mathbb{M}_{G} $. Since the length of path  $ p(i,~\xi)  + p(\xi,~j)$ is great than or equal to $ \mathbb{M}_{G}[i,~\xi]  + \mathbb{M}_{G}[\xi,~j] $, which is contradict to Step 9 of Algorithm \ref{alg:sp_apsp}, which says the shortest path distance from node $ i $ to  $ j $ on graph $ G $ is less than $ \mathbb{M}_{G}[i,~\xi]  + \mathbb{M}_{G}[\xi,~j] $. So,  a shortest path from node $ i $ to  $ j $ on graph $ G $ cannot involve a node in $G_r$, the correctness of Theorem \ref{al_correct} is proved.
\end{proof}

\subsection{All shortest paths between two nodes}
The generated $small\_matrix$ and $ candidate\_node\_list $ in Algorithm \ref{alg:sp_apsp} can be used to calculate all the shortest paths between two nodes on graph $ G $. Algorithm \ref{alg:sp_apsp_all} is devised for warm-start calculation of all the shortest paths between two nodes, based on the APSP matrix and conclusion of Theorem \ref{al_correct}.

We can even enumerate all the paths from node $ i $ to $ j $ to check if it is a shortest path, since the graph decided by $small\_matrix$ is small.
 
\begin{algorithm} 
	\caption{warm-start calculation of all shortest paths}
	\begin{algorithmic}[1]
	\Require{$small\_matrix$: $ \mathbb{M}_{s} $,  $candidate\_node\_list$: $ \mathbb{C}_{l} $,  start node: $ i$,  end node: $ j$}
		\Ensure{All shortest paths from $ i $ to $ j $ on graph $ G $: $ \mathbb{P}(i,~j) $}
		\Statex
		\Function{warm\_cal\_all\_shortest\_paths}{$ \mathbb{M}_{s} $, $ \mathbb{C}_{l} $, $ i$, $j$}
		\If{$ i == j $ }
		\State \Return {$[[i]] $}
		\EndIf

		\State {$ \mathbb{P}(i,~j) $ $\gets$ $empty\_list$}
		
		\State {Use Dijkstra's algorithm to calculate a shortest path from node $ i $ to $ j $, on the graph defined by $ \mathbb{M}_{s} $, noted $ \Psi_{(i,j)} $;}
 
		\State {Append $ \Psi_{(i,j)} $ to $ \mathbb{P}(i,~j) $;}
		
		\State {Divide  nodes  in $ \mathbb{C}_{l} $  into two sets: nodes in  $\{i,~j\} $, noted $\Phi_p$;  nodes not in  $\{i,~j\} $, noted  $\Phi_r$;}
	
		\For{$t$ $ in $ $\Phi_r$}
		\State {Calculate the shortest path from node $ i $ to  $ t $, and $ t $ to  $ j $, link the two paths into a new path $ P\_new $;}
		\State {Check if  $ P\_new $ is already in $ \mathbb{P}(i,~j) $, if yes, $ continue $;}
		\State {Append $P\_new$ to $ \mathbb{P}(i,~j) $;}				
		\EndFor

		\State \Return {$ \mathbb{P}(i,~j) $}
		\EndFunction
	\end{algorithmic}
	\label{alg:sp_apsp_all}
\end{algorithm}

\subsection{All shortest paths on undirected graph}
When the graph is undirected and the APSP matrix is unknown, Algorithm \ref{alg:all_sp_undirected} can be used to calculate all shortest paths between two nodes. Since the APSP matrix is unknown, the calculation is cold-start. The average case complexity of Algorithm \ref{alg:all_sp_undirected} is $ O(n^2) $. When the graph is directed, the complexity is $O(n^3) $, because we need $O(n^3) $ time to calculate the APSP matrix firstly.

\begin{algorithm} 
	\caption{Cold-start calculation of all shortest paths on undirected graph}
	\begin{algorithmic}[1]
	\Require{Adjacency matrix: $ \mathbb{A}_{G} $,  start node: $ i$,  end node: $ j$}
		\Ensure{All shortest paths from $ i $ to $ j $ on graph $ G $: $ \mathbb{P}(i,~j) $}
		\Statex
		\Function{all\_shortest\_paths\_undirected\_graph}{$ \mathbb{A}_{G} $, $ i$, $j$}
		\If{$ i == j $ }
		\State \Return {$[[i]] $}
		\EndIf
		
		\State {$remaining\_node\_list$ $\gets$ $G - i - j$}
		\State {$candidate\_node\_list$ $\gets$ $empty\_list$}
		\State {$candidate\_node\_list.append(i)$}
		
		\State {Use Dijkstra's algorithm to calculate shortest path distances from node $ i $ to all nodes on graph $ G $, noted $ \mathbb{V}_{i} $;}
		\State {Use Dijkstra's algorithm to calculate shortest path distances from node $ j $ to all nodes on graph $ G $, noted $ \mathbb{V}_{j} $;}		
 
		\For{$t$ $ in $ $remaining\_node\_list$}
		\If{$ \mathbb{V}_{i}[j]  <  \mathbb{V}_{i}[t]  + \mathbb{V}_{j}[t] $ }
		\State {$pass$}
		\Else
		\State {$candidate\_node\_list.append(t)$}
		\EndIf			
		\EndFor	
		\State {$candidate\_node\_list.append(j)$}		
		
		\State {Calculate $small\_matrix$  $ \mathbb{M}_{s} $ with $candidate\_node\_list$ and $ \mathbb{A}_{G} $;}

		\State {Use Algorithm \ref{alg:sp_apsp_all} to calculate all shortest paths between $ i $ and $ j $ on graph $ G $, noted $ \mathbb{P}(i,~j) $;}

		\State \Return {$ \mathbb{P}(i,~j) $}
		\EndFunction
	\end{algorithmic}
	\label{alg:all_sp_undirected}
\end{algorithm}

\subsection{Maintaining a  key\_node\_list }
When all shortest paths from node $ i $ to $ j $ is known, we can calculate a $ key\_node\_list $ for node pair $ (i,~j) $, which collects all the essential nodes to form a shortest path from node $ i $ to $ j $. When needing to remove a node, we can just check each pair of nodes' $ key\_node\_list $ to decide if the shortest path is affected. Algorithm \ref{alg:APSPremove_key} is a variant of Algorithm \ref{alg:APSPremove}, which calculates the new APSP matrix by utilizing the $ key\_node\_list $. Since the $ key\_node\_list $ for each pair of nodes is usually small, the average case space complexity is $ O(n^2) $. 

Step 2 to 8 of Algorithm \ref{alg:APSPremove_key} can be calculated in advance of knowing which node is about to be removed. Algorithm \ref{alg:APSPremove_key}  works even when all shortest paths calculated in Step 4 is not complete (e.g., we have missed some shortest paths during Step 4).

\begin{algorithm} 
	\caption{warm-start calculation of APSP by $ key\_node\_list $}
	\begin{algorithmic}[1]
		\Require{$G$, APSP of G: $ \mathbb{M}_{G} $,  node to be removed: $ G_k $, hyper-parameter: $\delta$} 
		\Ensure{$APSP ~of ~G - G_k:  \mathbb{M}_{G - G_k} $}
		\Statex
		\Function{APSP\_by\_key\_node\_list}{$G$, $ \mathbb{M}_{G} $, $ G_k $, $\delta$}
		
		\State {$key\_node\_list\_all$ $\gets$ $empty\_list$}
		\For{Each pair of node $ (i,~j) $}
		\State {Use Algorithm \ref{alg:sp_apsp_all} to calculate all the shortest paths from node $ i $ to $ j $ on  graph $ G $;}
		\State {Calculate the intersection of all shortest paths, noted $ \mathbb{L}_{k} $, $ \mathbb{L}_{k} $ collects all the essential nodes to form a shortest path from node $ i $ to $ j $;}	
		\State {Remove node $ i $ and $  j $  from $ \mathbb{L}_{k} $;}
		\State {$key\_node\_list\_all.append(\mathbb{L}_{k})$}
		\EndFor	
		
		\State {Calculate the $need\_update\_list$ of Algorithm \ref{alg:APSPremove} with $key\_node\_list\_all$, by checking each pair of nodes' $\mathbb{L}_{k}$ to decide if the shortest path is affected when removing  $ G_k $;}	
		
		\State {Use Step $ 4 $ to $ 20 $ of Algorithm \ref{alg:APSPremove} to calculate $\mathbb{M}_{G - G_k} $;}
 
		\State \Return {$\mathbb{M}_{G - G_k} $}
		\EndFunction
	\end{algorithmic}
	\label{alg:APSPremove_key}
\end{algorithm}

\subsection{Another variant of Algorithm \ref{alg:APSPremove}}
Although it can be calculated in advance of knowing which node is to be removed, the $key\_node\_list\_all$ in Algorithm \ref{alg:APSPremove_key} is very expensive to calculate, the time complexity is at least $ O(n^3) $. Therefore, we devise another variant of Algorithm \ref{alg:APSPremove}, which uses the conclusion of Corollary \ref{corollary3} and the technique used in Algorithm \ref{alg:APSPremove_key} to calculate the $key\_node\_list$, for one pair of nodes. Then check if the being removed node is in the $key\_node\_list$ from node $ i $ to $ j $. To save some time, we replace $\{i,~j\} $ with $ \Psi_{(i,j)} $ in Step 8 of Algorithm \ref{alg:sp_apsp_all}.
 
The new algorithm is referred to as Algorithm \ref{alg:APSPremove_key_new}. Further experiment in Section \ref{sec:test} shows Algorithm \ref{alg:APSPremove_key_new} performs better than Algorithm \ref{alg:APSPremove}.

\begin{algorithm} 
	\caption{APSP by $ key\_node\_list $ and Corollary \ref{corollary3}}
	\begin{algorithmic}[1]
		\Require{Adjacency matrix: $ \mathbb{A}_{G} $, APSP of G: $ \mathbb{M}_{G} $,  node to be removed: $ G_k $, hyper-parameter: $\delta$} 
		\Ensure{$APSP ~of ~G - G_k:  \mathbb{M}_{G - G_k} $}
		\Statex
		\Function{APSP\_by\_key\_node\_list}{$G$, $ \mathbb{M}_{G} $, $ G_k $, $\delta$}
		\State {$remaining\_node\_list$ $\gets$ $G - G_k$}
		\For{$i$ $ in $ $remaining\_node\_list$}
		\For{$j$ $ in $ $remaining\_node\_list$}
		\If{$ \mathbb{M}_{G}[i,~j]  >=  \mathbb{M}_{G}[i,~k]  + \mathbb{M}_{G}[k,~j] $ }
		\State {Calculate the $ key\_node\_list $ from $ i $ to $ j $ with the technique used in Algorithm \ref{alg:APSPremove_key}; }
		\If{$ G_k $ in $ key\_node\_list $}	
		\State {$need\_update\_list[i].append(j)$}
		\EndIf
		\EndIf	
		\EndFor 
		\EndFor 		
		\State {Use Step $ 4 $ to $ 20 $ of Algorithm \ref{alg:APSPremove} to calculate $\mathbb{M}_{G - G_k} $;}
 
		\State \Return {$\mathbb{M}_{G - G_k} $}
		\EndFunction
	\end{algorithmic}
	\label{alg:APSPremove_key_new}
\end{algorithm}

\subsection{A variant of Algorithm \ref{alg:APSPremove_key_new}}
Algorithm \ref{alg:APSP_algo12} is a variant of Algorithm \ref{alg:APSPremove_key_new}. It is based on the conclusions of Theorem \ref{theo_algo12} and Corollary \ref{corollary3}. Preliminary test shows Algorithm \ref{alg:APSP_algo12} is slightly faster than Algorithm \ref{alg:APSPremove_key_new}.

\begin{theorem}
	\label{theo_algo12}
Suppose $ i,~j,~k\in \{1,2, \dots, N \}$,	$ k \neq i$, $ k \neq j $. If the shortest path distance from node $ G_i $ to $ G_j $ on graph $ G - G_k $ is larger than on graph $ G $, then $ G_k $ is necessary for the shortest path from node $ G_i $ to $ G_j $ on graph $ G $.
\end{theorem}

\begin{proof}
	Suppose $ G_k $ is not necessary for the shortest path from node $ G_i $ to $ G_j $ on graph $ G $, which means there exists one shortest path from $ G_{i} $ to $ G_j $ which does not contain node $ G_{k} $,  on graph $ G $. That is to say: the shortest path distance from node $ G_i $ to $ G_j $ on graph $ G - G_k $ is equal to  on graph $ G $, which is contradicted to the condition that ``the shortest path distance from node $ G_i $ to $ G_j $ on graph $ G - G_k $ is larger than on graph $ G $."
\end{proof}

\begin{algorithm} 
	\caption{APSP by Theorem \ref{theo_algo12} and Corollary \ref{corollary3}}
	\begin{algorithmic}[1]
		\Require{Adjacency matrix: $ \mathbb{A}_{G} $, APSP of G: $ \mathbb{M}_{G} $,  node to be removed: $ G_k $, hyper-parameter: $\delta$} 
		\Ensure{$APSP ~of ~G - G_k:  \mathbb{M}_{G - G_k} $}
		\Statex
		\Function{APSP\_by\_key\_node\_list}{$G$, $ \mathbb{M}_{G} $, $ G_k $, $\delta$}
		\State {$remaining\_node\_list$ $\gets$ $G - G_k$}
		\For{$i$ $ in $ $remaining\_node\_list$}
		\For{$j$ $ in $ $remaining\_node\_list$}
		
		\If{$ \mathbb{M}_{G}[i,~j]  >=  \mathbb{M}_{G}[i,~k]  + \mathbb{M}_{G}[k,~j] $ }
		
		\State {$temp\_list$ $\gets$ $G - G_i - G_j - G_k$}
		\State {$candidate\_node\_list$ $\gets$ $empty\_list$}
		\State {$candidate\_node\_list.append(i)$}
		\For{$t$ $ in $ $temp\_list$}
		\If{$ \mathbb{M}_{G}[i,~j]  <  \mathbb{M}_{G}[i,~t]  + \mathbb{M}_{G}[t,~j] $ }
		\State {$pass$}
		\Else
		\State {$candidate\_node\_list.append(t)$}
		\EndIf			
		\EndFor	
		\State {$candidate\_node\_list.append(j)$}
		\State {$K$ $\gets$ $len(candidate\_node\_list)$}		
		\State {Use $candidate\_node\_list$ and $ \mathbb{A}_{G} $ to calculate the $small\_matrix$;}		
		\State {Use Dijkstra's algorithm to calculate the shortest path distance from node $ 0 $ to node $ K-1 $, on the graph defined by $small\_matrix$, noted $ \Phi_{ij} $;}	
		\If{$\Phi_{ij}   >   \mathbb{M}_{G}[i,~j] $}	
		\State {$need\_update\_list[i].append(j)$}
		\EndIf					 	
		\EndIf	
		\EndFor 
		\EndFor 		
		\State {Use Step $ 4 $ to $ 20 $ of Algorithm \ref{alg:APSPremove} to calculate $\mathbb{M}_{G - G_k} $;}
 
		\State \Return {$\mathbb{M}_{G - G_k} $}
		\EndFunction
	\end{algorithmic}
	\label{alg:APSP_algo12}
\end{algorithm}

\section{Testing of the  algorithms} \label{sec:test}
We tested the algorithms for warm-start calculation of the new APSP matrix after a minor update of a dense graph, e.g., removing a node, or modifying an edge. All the code in the experiments is implemented with Python. To compare the  algorithms more reliably, we convert the Python code into C++ code.

\subsection{Experiment \RNum{1}} 
In Experiment \RNum{1}, we test warm-start calculation of the new APSP matrix after removing a node, and compare with cold-start calculation of the Floyd-Warshall algorithm. In the experiment, a random node is removed from a complete graph, then record the time spent for calculating the new APSP matrix, by warm-start calculation of Algorithm \ref{alg:APSPremove} and cold-start of Floyd-Warshall algorithm. The ratio of the used time is calculate with Equation \ref{equ:ratioo}. Different sizes of complete graphs are tested, from 1,000 nodes to 5,000 nodes. For each size of graph less than 5,000 nodes, we repeat the experiment 20 times and calculate the average and standard deviation (SD) of  the ratios; for graph of size of 5,000 nodes, the experiment is repeated five times. 
\begin{equation}
r =  \frac{APSP_{warm}(G^{'})}{APSP_{cold}(G^{'})}	
\label{equ:ratioo}
\end{equation}
\begin{table}
		\caption{Warm-start vs. cold-start calculation of the APSP matrix, after removing a node.}
	\begin{center}
		\scalebox{0.95}{
			
\begin{tabular}{llllll}
	\toprule
	& N = 1000 & N = 2000 & N = 3000 & N = 5000 & Avg \\ \hline
	Ratio & 0.58$ \pm $0.32& 0.48$ \pm $0.22& 0.61$ \pm $0.19& 0.6$ \pm $0.13& 0.57   \\ \bottomrule
\end{tabular} 	
}	
 
		\label{tab:exp1}
	\end{center}
\end{table}
As shown in Table \ref{tab:exp1}, a warm-start calculation only needs 0.57 of the time, to a  cold-start calculation with the Floyd-Warshall algorithm on average. That means we can save 43\% of calculation time if using warm-start calculation. 

\subsection{Experiment \RNum{2}} 
The setting of Experiment \RNum{2} is similar to Experiment \RNum{1}, except that we are testing modifying an edge, not removing a node. As shown in Table \ref{tab:exp2}, we can save 50\% of calculation time if using warm-start calculation, when compared with the Floyd-Warshall algorithm.
 
\begin{table}
		\caption{Warm-start vs. cold-start calculation of the APSP matrix, after modifying an edge.}
	\begin{center}
		\scalebox{0.95}{
			
\begin{tabular}{llllll}
	\toprule
	& N = 1000 & N = 2000 & N = 3000 & N = 5000 & Avg \\ \hline
	Ratio & 0.46$ \pm $0.24& 0.46$ \pm $0.21& 0.58$ \pm $0.23& 0.51$ \pm $0.11& 0.5   \\ \bottomrule
\end{tabular} 	
}		
 
		\label{tab:exp2}
	\end{center}
\end{table}

\subsection{Experiment \RNum{3}} 
In Experiment \RNum{3}, we test warm-start calculation of the shortest path between two nodes, based on the known APSP matrix and Theorem \ref{theorem_spd_4}. And compared with cold-start calculation of the shortest path by Dijkstra's algorithm.

Other settings of the experiment are similar to Experiment \RNum{1} and \RNum{2}. In the experiment, we test calculating the shortest path between two nodes, on complete graphs of different sizes, by warm-start and cold-start calculation separately. Each method is repeated 1,000 times.  The result shows  a warm-start calculation only needs 0.01 of the time of a  cold-start calculation on average. That means we can save 99\% of calculation time if using warm-start calculation. 
\begin{table}
		\caption{Warm-start vs. cold-start calculation of shortest path.}
	\begin{center}
		\scalebox{0.95}{
			
\begin{tabular}{llllll}
	\toprule
	& N = 1000 & N = 2000 & N = 3000 & N = 5000 & Avg \\ \hline
	Ratio & 0.02$ \pm $0.00& 0.02$ \pm $0.01& 0.01$ \pm $0.00& 0.007$ \pm $0.00& 0.01   \\ \bottomrule
\end{tabular} 	
}		
 
		\label{tab:exp3}
	\end{center}
\end{table}
 
 \subsection{Experiment \RNum{4}}
 The setting of Experiment \RNum{4} is similar to Experiment \RNum{1}, except that we are using Algorithm \ref{alg:APSPremove_key_new}, instead of Algorithm \ref{alg:APSPremove}. By comparing Table \ref{tab:exp4} with Table \ref{tab:exp1}, we can see that Algorithm \ref{alg:APSPremove_key_new} performs better than Algorithm \ref{alg:APSPremove}.
 \begin{table}
		\caption{Warm-start vs. cold-start calculation of the APSP matrix, after removing a node, by Algorithm \ref{alg:APSPremove_key_new}.}
		\scalebox{0.95}{
			
\begin{tabular}{llllll}
	\toprule
	& N = 1000 & N = 2000 & N = 3000 & N = 5000 & Avg \\ \hline
	Ratio & 0.45$ \pm $0.31& 0.34$ \pm $0.22& 0.24$ \pm $0.13& 0.33$ \pm $0.16& 0.34   \\ \bottomrule
\end{tabular} 	
}		
	\label{tab:exp4}
\end{table}

\section{Discussion} 
The algorithms can be revised for warm-start calculation of the minimax path problem or widest path problem, on a large dense graph.

\section{Conclusion} 
We propose two algorithms (and some variants) for  warm-start calculation of the all-pairs shortest path (APSP) matrix after a minor modification of a weighted dense graph, e.g., adding a node, removing a node, or updating an edge. We assume the APSP matrix for the original graph is already known, and try to warm-start from the known APSP matrix to reach the new APSP matrix.  A cold-start calculation of the APSP matrix for the updated graph needs $ O(n^3) $ time. It is a very expensive time cost for a large dense graph. We are trying to utilize the already calculated APSP matrix to make calculation of the new APSP matrix less expensive. The best case complexity for a warm-start calculation is $ O(n^2) $, the worst case complexity is $ O(n^3) $.

We implemented the algorithms and tested their performance with experiments. The result shows a warm-start calculation can save a large portion of calculation time when compared with the Floyd-Warshall algorithm. Moreover, we proposed another algorithm for warm-start computing of the shortest path between two nodes, and tested it. Result shows warm-start computing can save 99\% of time, compared with cold-start computing  of the shortest path by Dijkstra's algorithm.

\bibliographystyle{ACM-Reference-Format}
\bibliography{mmj} 
 
\end{document}